\documentclass[11pt]{article}
\usepackage{graphicx,comment}
\usepackage{geometry}[1in]
\usepackage{amsthm,amsmath,amssymb}
\usepackage{hyperref,color}
\usepackage[capitalize,nameinlink]{cleveref}
\usepackage[dvipsnames]{xcolor}
\hypersetup{
	colorlinks=true,
	pdfpagemode=UseNone,
    citecolor=OliveGreen,
    linkcolor=NavyBlue,
    urlcolor=Magenta,
	pdfstartview=FitW
}
\usepackage{appendix}

\newtheorem{theorem}{Theorem}
\newtheorem{definition}[theorem]{Definition}
\newtheorem{lemma}[theorem]{Lemma}

\newcommand{\emm}{\mathrm{e}}
\newcommand{\eps}{\varepsilon}

\newcommand\norm[1]{\left\|{#1}\right\|}

\newcommand{\pl}{\textup{\texttt{+}}}
\newcommand{\mi}{\textup{\texttt{-}}}
\newcommand{\plm}{\textup{\texttt{\textpm}}}

 \author{
Andreas Galanis\thanks{Department of Computer Science, University of Oxford. Email:
   \texttt{andreas.galanis@cs.ox.ac.uk}.
 }
 \and
 Daniel \v{S}tefankovi\v{c}\thanks{Department of Computer Science, University of Rochester. Email: \texttt{stefanko@cs.rochester.edu}. Research supported in part by NSF grant CCF-1563757.}
 \and 
 Eric Vigoda\thanks{Department of Computer Science, University of California, Santa Barbara. Email: \texttt{vigoda@ucsb.edu}. Research supported in part by NSF grant CCF-2147094.}
 }
\title{Critical window for approximate counting\\ in dense Ising models}
\date{\today}

\begin{document}
\maketitle
\begin{abstract}
We study the complexity of approximating the partition function of dense Ising models in the critical regime. Recent work of Chen, Chen, Yin, and Zhang (FOCS 2025) established fast mixing  at criticality, and even beyond criticality  in a window of width $N^{-1/2}$. We complement these algorithmic results by proving nearly tight hardness bounds, thus yielding the first instance of a sharp scaling window for the computational complexity of approximate counting.

Specifically, for the dense Ising model we show that approximating the partition function is computationally hard within a window of width \(N^{-1/2+\varepsilon}\) for any constant \(\varepsilon>0\). Standard hardness reductions for non-critical regimes break down at criticality due to bigger fluctuations in the underlying gadgets, leading to suboptimal bounds. We overcome this barrier via a global approach which aggregates fluctuations across all gadgets rather than requiring tight concentration guarantees for each individually. This new approach yields the optimal exponent for the critical window.\let\thefootnote\relax\footnotetext{For the purpose of Open Access, the authors have applied a CC BY public copyright licence to any Author Accepted Manuscript version arising from this submission. All data is provided in full in the results section of this paper.}
\end{abstract}
\section{Introduction}

Phase transitions are often accompanied by \emph{critical windows}: narrow parameter regimes around a critical point in which the global behavior of the system changes rapidly as a function of the system size. The delicate behavior inside such windows has been a key challenge in probability and statistics, and has recently emerged as a new frontier in sampling and approximate counting. In this paper, we study critical windows for the Ising model, focusing on their implications for computational complexity.

Given a symmetric interaction matrix \(J \in \mathbb{R}^{N \times N}\), the Ising model is the probability distribution \(\mu_J\)  defined by $\mu_J(\sigma) \propto \exp\big(\tfrac{1}{2}\sigma^\top J \sigma\big) \mbox{ for all } \sigma\in \{\mi1,\pl1\}^N$. We will refer to the entries of $J$ as  weights. The normalizing constant $Z_J = \sum_{\sigma \in \{\pl1,\mi1\}^N} \exp\!\left(\frac{1}{2}\sigma^\top J \sigma\right)$
is called the partition function.  Typically one takes $J$ to be the adjacency matrix of a graph $G$ scaled by some parameter, and in this case $\mu_J$ can also be viewed as a weighted distribution on the cuts of the graph. One of the most well-studied cases is the $N$-vertex complete graph where each edge has weight $\frac{\beta}{N}$ and $\beta>0$ is known as the inverse temperature. We will focus on the setting where $J$ is positive semidefinite (psd) and its spectral norm is bounded by some constant, specifically  when  $\norm{J}_2\approx 1$ which captures criticality for dense graphs.\footnote{\label{fn:psd}Note that by adding to $J$ a scalar multiple of the identity matrix, the Ising distribution does not change;  so having $J$ to be psd and $\norm{J}_2\approx 1$  is equivalent to considering matrices $J$ whose spectral diameter (the difference between the largest and smallest eigenvalues of $J$) is roughly equal to 1.} 

In the case of the complete graph, it is well-known that the model undergoes a phase transition as $N$ grows large at the critical point $\beta_c=1$. For a configuration $\sigma\in \{\mi 1,\pl 1\}^N$, let   $m_\sigma=\sum_{i\in [N]} \sigma_i$ be the so-called magnetization of $\sigma$. When $\beta<\beta_c$, entropy dominates and correlations are weak, so a typical sample $\sigma\sim\mu_J$ is roughly balanced between $\plm 1$ spins and satisfies $|m_\sigma|=O( N^{1/2})$. For $\beta>\beta_c$,  interactions dominate and long-range order emerges in the sense that a typical $\sigma\sim \mu_J$ is biased towards either the ``$\pl$ phase'' or the ``$\mi$ phase'' where  $|m_\sigma|=\Theta(N)$. So, away from criticality, the underlying system behaves predictably. Around criticality however, for $\beta=1\pm O(\tfrac{1}{N^{1/2}})$, the above high and low temperatures principles no longer apply cleanly; instead they lead to a conflating behaviour which exhibits nontrivial scaling:  a typical sample $\sigma\sim \mu_J$ has  still a somewhat balanced number of $\plm 1$ but $|m_\sigma|=O( N^{3/4})$, i.e., the fluctuations of the spins are of order $N^{3/4}$. 

Understanding critical-window regimes has been a key challenge in  approximate counting/sampling, where  there is a vast amount of results for non-critical cases but only few that apply to criticality, especially for more general classes of graphs. In terms of Markov chains, it has been shown on the complete graph \cite{glaubera, glauberb} that the standard Glauber dynamics has mixing time $O(N\log N)$ for $\beta<1$, $O(N^{3/2})$ for $\beta=1$, and $\emm^{\Omega(N)}$ for $\beta>1$, see also~\cite{long2014power, cuff2012glauber,GSV-SwendsenWang,blanca2015dynamics,blanca2025mean,cutoffSW,helmuth2023finite} for related results. Similar flavored but more delicate-to-prove Markov-chain results have also been developed in the 2D lattice~\cite{martinelli19942,lubetzky2012critical,gheissari2018mixing}.

From an algorithmic/complexity viewpoint, progress for the case of a general interaction matrix $J$ has only been made very recently, using spectral independence theory~\cite{ALO20}.   In particular, for the case of the Ising model with general interaction matrix $J$, \cite{Ja,Jb} showed that the Glauber dynamics mixes rapidly when $\norm{J}_2<1$, whereas \cite{Inapproxa, Inapproxb} showed hardness of approximately sampling/counting  when $\norm{J}_2>1$. Note that the condition $\norm{J}_2=1$ translates precisely to the critical value $\beta_c=1$ for the complete graph. An analogous line of techniques has been developed for sparse Ising models (e.g. on graphs of bounded degree $d$), where the qualitative picture is analogous but the relevant phase transitions are more closely related to regular trees and random regular graphs. 

Despite these developments for non-critical regimes, understanding the complexity at criticality appeared to still be elusive. Recent work of Chen, Chen, Yin and Zhang \cite{CCYZ-critical-hard-core} however made a major leap forward, showing fast mixing for the Glauber dynamics when  $\|J\|_2=1$, and showed similar results for the Ising and hard-core model in sparse graphs. For the dense setting in particular that we consider here, they managed to go even beyond criticality and showed that, for any fixed constant $\alpha>0$, the mixing time is $O(N^{3/2}\log{N})$ when $\|J\|_2 \leq 1 + \frac{\alpha}{N^{1/2}}$. 

We prove a nearly-matching hardness result, thereby establishing the first example of a sharp scaling window for the computational complexity of the approximation problem.

\begin{theorem}\label{thm:Ising}
Fix any arbitrarily small constant $\varepsilon>0$. 
It is NP-hard to approximate the Ising partition function on $N\times N$ psd matrices $J$ when $\|J\|_2 \leq  1 + \frac{1}{N^{1/2-\varepsilon}}$, even within a factor $\exp(N^c)$ where $c=c(\varepsilon)>0$ is a constant depending only on $\varepsilon$.
\end{theorem}

We remark that away from criticality it is sometimes possible to obtain hardness of approximation within a factor of $\exp(\Omega(N))$. This is \emph{not} the case for the critical case we are considering here: for any $\eta>0$ and $\|J\|_2 = 1 + \eta$, setting $\hat{J}=\tfrac{1}{1+\eta}J$ yields that $\|\hat{J}\|_2=1$ and $Z_J/Z_{\hat{J}}=\exp(\Theta(N\eta))$. One can then obtain in poly-time an approximation to $Z_{\hat{J}}$, and therefore $Z_J$, using the critical-case algorithm of \cite{CCYZ-critical-hard-core}. For $\eta=\tfrac{1}{N^{1/2-\varepsilon}}$, this yields in particular an $\exp(\Theta(N^{1/2+\varepsilon}))$-factor approximation to $Z_J$, even for matrices $J$ with $\|J\|_2 = 1 + \frac{1}{N^{1/2-\varepsilon}}$ that fall into the scope of Theorem~\ref{thm:Ising}.

For the sparse Ising and hard-core models, Chen, Chen, Yin and Zhang~\cite{CCYZ-critical-hard-core} (see also \cite{ChenJiang}) establish polynomial mixing of the Glauber dynamics at criticality for any graph of maximum degree $\Delta$.  Moreover, these works conjecture that the polynomial mixing extends to a window of width $N^{-1/2}$, but their current proof approaches do not extend beyond the critical point, and hence do not establish a scaling window.  Therefore, establishing a scaling window for algorithmic or hardness results in the sparse regime is an intriguing open problem.

\section{Proof Overview}
For the purposes of this section, we will assume that $\norm{J}_2=1+\eta$ for some small $\eta>0$. Before giving the proof of our inapproximability result, it is useful to explain the reduction scheme from \cite{Inapproxa} for the non-critical case where $\eta$ is a (small) constant and where it runs into problems in our case when we allow $\eta$ to depend on $N$. We remark that the idea for the reduction scheme goes back to the works of \cite{Dyer1,Sly2010, SlySun} for counting independent sets, and has since been the main template in follow-up works for establishing hardness of approximation for partition functions.

Key to the reduction is to exploit the phase transition on the complete graph that occurs for $\beta>\beta_c=1$ (which corresponds to $\norm{J}_2>1$) in order to encode \textsc{Max-Cut}. To align with the notation in the next section, let $t$ be the number of vertices in the graph. Recall, in this regime, a typical configuration $\sigma\in \{\mi 1,\pl 1\}^t$ is either biased towards $\pl 1$ or $\mi 1$, in the sense that the magnetization $m_\sigma$ satisfies $m_\sigma=\Omega(t)$ with high probability over~$\sigma$. More precisely, for every constant $\beta>1$ there is  $b=b(\beta)>0$ so that
\[\Pr\big[m_\sigma=bt +o(t)\big]=\tfrac{1}{2}+o(1)\mbox{ and }\Pr\big[m_\sigma=-bt +o(t)\big]=\tfrac{1}{2}+o(1).\]  
The reduction uses this bimodality in order to encode \textsc{Max-Cut} on 3-regular graphs $G$. So, replace every vertex of $G$ with a copy of the $t$-vertex complete graph; for each edge of~$G$ place a disjoint matching of  size $r$ between the corresponding copies of the gadget and place negative weights $-\gamma$ on the edges for some small $\gamma>0$ so that the spectral norm of the final graph is dominated by that of the gadget.\footnote{\label{fn:bc}
Setting $\beta=1+\Theta(\eta)$ and $\gamma=\Theta(\eta)$ for example gives that  $\norm{J}_2\leq \beta+6\gamma=1+O(\eta)$.} The idea is then that each gadget can  be either in the $\pl$ or $\mi$ phase (thus encoding a cut partition), whereas the negative weights on the $r$ edges between neighboring gadgets favour the phases to be different. More precisely,  for every phase assignment $Y=\{\mi,\pl\}^V$, the total contribution to the partition function of configurations that are consistent with the phase assignment is proportional to $c^{r|\mbox{\footnotesize\textsc{Cut}}(Y)|}$ where  for small $\gamma$ one has  $c=1+2(2b-1)^2\gamma +o(\gamma^2)$ (this can be derived, for example, from the argument in \cite[Appendix A.1]{Inapproxa}). Crucially, for fixed $\beta>1$, $c$ is a constant bounded away from 1 allowing to make the contribution of the max-cut dominant by taking the number of connections $r$ sufficiently large.

There are two bottlenecks in this argument that work against each other: \begin{enumerate}
    \item The connections between the gadgets need to amplify the phase disagreement so that the contribution of the max-cut prevails over  suboptimal cuts; in practice, this requires that the number $r$ of connections between gadgets is big enough so that $c^r=1+\Omega(1)$.
    \item On the other hand, the number of connections between gadgets needs to be small enough so that the fluctuations from the connections does not interfere with the phases of the gadget. In fact, to pull through the calculation with the weight of phase assignment, the endpoints of the connections within every gadget were required to behave  independently (conditioned on the phase). This restricted the number of ports to be a small polynomial relative to the total size of the gadget. For non-critical cases for example $r=o(\sqrt{t})$ was absolutely necessary in order to guarantee this conditional independence.
\end{enumerate} For non-critical cases, the bias $b$ towards the phase is bounded away from 0, leading to $c$ being a constant bigger than one, allowing therefore to take $r$ as small as a constant. For the critical case however, these bottlenecks become far trickier to balance: when $\beta= 1+\eta$ for some small $\eta$, $b$ is roughly $\tfrac{1}{2}+\Theta(\sqrt{\eta})$, and with the choice $\gamma=\Theta(\eta)$ (cf. Footnote~\ref{fn:bc}), this leads to $c$ being $1+\Theta(\eta^2)$. In order to ensure that $c^r=\Omega(1)$  it is therefore required that $r>1/\eta^2$. For $\eta=t^{-1/2+\varepsilon}$ this already conflicts the requirement 2) quite strongly. The best bound that one can obtain along these lines is inapproximability for $\norm{J}_2\leq 1+N^{-1/4+\varepsilon}$, and even that requires surprisingly careful arguments in part 2) due to the weaker concentration because of the criticality (the decay is quartic rather than quadratic around the mode). So, to obtain a tight window, we need a fundamentally different approach.

 Our approach to obtaining the optimal window is surprisingly direct. We give up completely on maintaining/proving the strong form of conditional independence that was used in previous reductions (mentioned in Item 2 above). Freed from this restriction, we can now use the overall bias of the gadget to control its interaction with the other gadgets and hence get the amplification of the phase disagreement by placing complete bipartite graphs between the gadgets. The price we have to pay is that the fluctuations coming from the gadget can potentially be enormous (of $O(t)$) and might shift the contribution of a cut unexpectedly. 

Fortunately, we can capture what is happening using a more global perspective. Namely, for each gadget we keep track of the bias $b_v$ of the gadget; for each vector $\{b_v\}_{v\in V}$, we can write an explicit analytic expression $\Phi$ that captures the contribution to the partition function. The goal then is, roughly, to show that the signs of $b_v$ in the dominant vector yield a large cut with large value. One of the key steps (see Lemma~\ref{lo}) is to show that at the optimum, each $b_v$ is actually bounded by $t^{3/4+O(\varepsilon)}$, rather than the worst-case $O(t)$ mentioned above; note that this matches the behaviour of the gadget that one expects at critical regimes. Having this in place, the rest of the proof is somewhat more streamlined, requiring only some careful estimates to recover the right asymptotics and control the gap between near-optimal and suboptimal cuts.

\section{Hardness Proof}

As outlined in the previous section, the starting point of the reduction is the result that it is APX-hard to approximate within a constant factor the size of the MAX-CUT for $3$-regular graphs~\cite{alimonti1997hardness}. 
Let $G=(V,E)$ be a $3$-regular graph. Let $n=|V|$ and $t$ a positive even integer to be fixed later (we will take $t=n^C$ for a sufficiently large constant $C$). 

\begin{definition}\label{def1}
Let $\beta, \gamma$ be positive reals. For a graph $G=(V,E)$ and even integer $t$, let $G_{t}(\beta,\gamma)=(V\times [t],E')$ where in $E'$ we have
\begin{itemize}\itemsep0em 
\item an edge of weight $\beta$ between $(v,i)$ and $(v,j)$ for every $v\in V$ and every $i\neq j\in [t]$,
\item an edge of weight $-\gamma$ between $(u,i)$ and $(v,j)$ for every $\{u,v\}\in E$ and every $i,j\in [t]$.
\end{itemize}
We define the ``cloud'' of $v$ to consist of vertices $(v,1),\dots,(v,t)$.
\end{definition}

Consider the Ising model on $G_t=G_t(\beta,\gamma)$ with the above weights. Observe that the interactions within each cloud are positive and the interactions across clouds are negative.  Moreover, there is a natural correspondence between Ising configurations on $G_t$ and cuts on $G$: we let the corresponding cut contain the vertices $v$ 
where the sum of the spins in the cloud of $v$ is $\leq 0$ (the vertices on the other side of the cut have sum of spins $>0$; the minor asymmetry in the correspondence is not important since $t$ will be large). 

For any $\eps>0$ and $\tau>1$, we will show that there exists a choice of $\beta, \gamma$ and $t$ such that 
1) an efficient algorithm to approximate the partition function of the Ising model on $G_t$  yields an efficient algorithm to obtain a $\tau$-approximation of the max-cut of $G$, and 
2) the $N\times N$ (where $N=nt$) interaction matrix $J$ of the Ising model on $G_t$ has spectral diameter $\leq  1 + N^{-1/2+\eps}$ (recall that the spectral diameter is the max minus the min eigenvalue of $J$, cf. Footnote~\ref{fn:psd}). More precisely, we will establish the following.

\begin{theorem}\label{TMAIN}
For any $\eps>0$ and $\tau>1$, there are $c,C>0$ so that for all  sufficiently large~$n$, there exist $\beta,\gamma>0$ such that the following holds. 

For any 3-regular graph  $G=(V,E)$ with $n$ vertices, let $G_t=G_t(\beta,\gamma)$ be the Ising instance from Definition~\ref{def1} with $t=n^C$, and $J$ be its $N\times N$ interaction matrix, where $N=n t = n^{C+1}$ is the number of vertices in $G_t$. Then, the spectral diameter of  $J$ is $\leq 1 + N^{-1/2+\eps}$. Moreover, for any $A$ with  $A\geq \tau |E|/2$,   there are $T_1,T_2>0$ with $T_1/T_2\geq \exp(2N^c)$ such that:\vspace{-0.15cm}
\begin{enumerate}
\itemsep0em 
\item for any $G$ with max-cut of size $\geq A$, the partition function $Z_J$ is $\geq T_1$, 
\item for any $G$ with max-cut of size $\leq \tfrac{A}{\tau}$, the partition function $Z_J$  is $\leq T_2$.
\end{enumerate}\vspace{-0.1cm}
The numbers $\beta,\gamma,T_1,T_2$ can be computed in $poly(n)$ time.
\end{theorem}

Theorem~\ref{TMAIN} immediately yields Theorem~\ref{thm:Ising}.

\begin{proof}[Proof of Theorem~\ref{thm:Ising}.]
Fix arbitrary $\varepsilon>0$.  From the APX-hardness result of \cite{alimonti1997hardness}, there exists an explicit constant $\tau>1$ so that the following  gap-version of \textsc{Max-Cut} is \mbox{NP-hard}. The input is  a $3$-regular graph $G=(V,E)$ with $n$ vertices and a number $A$ with  $A\geq \tau |E|/2=3\tau n/4$. We are promised that 
$G$ either: (a) has a max-cut of size $A$, or (b) has all cuts of size $\leq A/\tau$; the goal is to decide which of the two cases holds. Let $C,c>0$ be the constants from Theorem~\ref{TMAIN}. We will show how to decide the gap version of \textsc{Max-Cut} using a poly-time algorithm that, on input an $N\times N$ psd matrix $J$ with $\|J\|_2 \leq  1 + \frac{1}{N^{1/2-\varepsilon}}$, returns an approximation of the partition function $Z_J$  within a factor $\exp(N^c)$.

In particular, on input $G$ and $A$, we let $t=n^C$ and $N=nt=n^{C+1}$ and compute the numbers $\beta,\gamma,T_1,T_2>0$ as in Theorem~\ref{TMAIN}. Then,  construct the instance $G_t=G_t(\beta,\gamma)$ of the Ising problem given by Definition~\ref{def1}. By Theorem~\ref{TMAIN}, the interaction matrix $J$ of $G_t$ has dimension $N\times N$ and spectral diameter $\leq 1 + N^{-1/2+\eps}$. Setting $\hat{J}=J-\lambda I$ where $\lambda$ 
is the smallest eigenvalue of $J$ yields that $\hat{J}$ is psd and has $2$-norm equal to the spectral diameter of $J$; we also have that $Z_{\hat{J}}=K Z_J$ where $K:=\exp(-\tfrac{1}{2}\lambda N)$. Note that $\lambda$ (and hence $K$) can be computed to sufficient precision in $poly(n)$ time.

Suppose we can approximate the partition function $Z_{\hat{J}}$ of the Ising model on $G_t$ with factor better than $R=\exp(N^c)$, i.e., obtain $\hat{Z}$ so that   $Z_{\hat{J}}/\hat{Z}\in[1/R,R]$. Then, $K \hat{Z}$ is an estimate of the partition function $Z_J$ with factor better than $R$. Moreover, using that $T_1/T_2\geq \exp(2N^c)$, we have that exactly one of $K \hat{Z}/R\geq T_1$ and $T_2\geq K \hat{Z} R$ holds.  In the former case, we can conclude that $Z_J\geq T_1$, while in the latter we can conclude that $T_2\geq Z_J$. That is, we can distinguish between the cases (a) and (b) 
for $G$ (since we can distinguish between the cases 1) and 2) for $G_t$ in the statement of Theorem~\ref{TMAIN}).
\end{proof}

We start by introducing notation needed for the proof of Theorem~\ref{TMAIN}. Let $G=(V,E)$ be a 3-regular graph with $n$ vertices. Let also $t$ be an even integer and $\beta,\gamma>0$ be real numbers, to be specified later. Consider finally the graph $G_t=G_t(\beta,\gamma)$ from Definition~\ref{def1} and let $J$ be its interaction matrix.

Given a vector $\mathbf{b}=\{b_v\}_{v\in V}$ (where $|b_v|\leq t/2$ for $v\in V$) consider a configuration $\sigma$ that has $t/2 + b_v$
vertices with spin $\pl1$ on the cloud of $v$, so that the remaining $t/2 - b_v$ vertices have spin $\mi1$. The number of
such configurations  is
\begin{equation}\label{ee1}
\prod_{v\in V} {t\choose t/2 + b_v},
\end{equation}
since we can choose the vertices with spin $\pl 1$ in each cloud independently of the others.

In the cloud for $v$ the number of monochromatic edges in $\sigma$ is
$$
{t/2 + b_v\choose 2} + {t/2 - b_v\choose 2} = 2{t/2\choose 2} + b_v^2,
$$
whereas the number of bichromatic is $(t/2+b_v)(t/2-b_v)$; their difference is  $2b_v^2-t/2$. Similarly, the number of monochromatic edges between clouds for $u$ and $v$ in $\sigma$ is
\[
(t/2+b_v)(t/2+b_u) + (t/2-b_v)(t/2-b_u) = t^2/2 + 2 b_u b_v,\]
whereas the number of bichromatic is $t^2/2 - 2 b_u b_v$; their difference is therefore $4b_u b_v$. It follows  that $\exp(\tfrac{1}{2}\sigma^T J\sigma)$ is equal to 
\begin{equation}\label{ee2}
K\exp\Big( 2\beta \sum_{v\in V} b_v^2 - 4 \gamma \sum_{\{u,v\}\in E} b_u b_v \Big),
\end{equation}
where $K:=\exp(-\tfrac{1}{2}\beta n t)$ is independent of $\sigma$. The partition function $Z_J$ is then equal to
\begin{equation}\label{eq:partitionfun}
Z_J = K\sum_{\mathbf{b}} \bigg( \prod_{v\in V} {t\choose t/2 + b_v}\bigg) \exp\Big( 2\beta \sum_{v\in V} b_v^2 - 4 \gamma \sum_{\{u,v\}\in E} b_u b_v \Big),
\end{equation}
where the sum ranges over all vectors $\mathbf{b}=\{b_v\}_{v\in V}$ with integer entries in $[-t/2,t/2]$. 

For $x\in [0,1]$, let $H(x)=- (x\ln x + (1-x)\ln (1-x))$ be the entropy function, with the convention that $0\ln 0\equiv 0$. For a real number $b\in [-t/2,t/2]$, let further
\begin{equation}\label{ZOB}
Q(b):= 2\beta b^2 + t\, H\big(\tfrac{1}{2} + \tfrac{b}{t}\big) \geq  2\beta b^2 + \log {t\choose t/2 + b},
\end{equation}
and for $\mathbf{b}=\{b_v\}_{v\in V}$ with entries in $[-t/2,t/2]$ let
$$
\Phi(\mathbf{b}) = \sum_{v\in V} Q(b_v) - 4\gamma \sum_{\{u,v\}\in E } b_u b_v.
$$
The function $\Phi(\mathbf{b})$ roughly captures the logarithm of the contribution of configurations
with specification $\{b_v\}_{v\in V}$ to the partition function, combining~\eqref{ee1} and~\eqref{ee2} 
(note that in light of the last inequality in~\eqref{ZOB} we have that $\Phi(\mathbf{b})$ is 
actually an upper bound on the logarithm of the contribution).

Let $\hat{b}\in (0,t)$. We are going to take
\begin{equation}\label{e1}
\beta = \frac{1}{4\hat{b}}\ln \frac{t + 2\hat{b}}{t - 2\hat{b}}.
\end{equation}
so that $\hat{b}$ is the maximizer of $Q(b)$ restricted to $b\geq 0$. Note that balanced configurations (those having equal number of plus and minus spins) have
\begin{equation}\label{Q0}
Q(0) = t\ln 2,
\end{equation}
and the configurations that (roughly) have maximum likelihood on a cloud have
\begin{equation}
\begin{split}
Q(\hat{b}) & = t\ln 2 - \frac{t}{2}\Big( (1+\hat{b}/t)\ln(1+2\hat{b}/t) + (1-\hat{b}/t)\ln(1-2\hat{b}/t) \Big) \\ \label{QB}
& = t\ln 2 + \frac{4}{3} \frac{\hat{b}^4}{t^3} + O\left(\frac{\hat{b}^6}{t^5}\right).
\end{split}
\end{equation}
In the reduction we will need to have a large (growing as a function of $t$) difference between the values of $Q(0)$ and $Q(\hat{b})$---this will
force us to take $\hat{b}=\omega(t^{3/4})$. (The intuition is that the function $Q$ describes the behavior of the Ising model on one cloud, the two maxima 
$Q(\hat{b})$ and $Q(-\hat{b})$ describe two phases, and the dip at $Q(0)$ separates the phases.)

\begin{lemma}\label{lem:maxb}
For any $\hat{b}\in (0,t)$ and $\beta=\frac{1}{4\hat{b}}\ln \frac{t + 2\hat{b}}{t - 2\hat{b}}$ as in~\eqref{e1}, the maximum of $Q(b)$ in \eqref{ZOB} over $b\in [0,t/2]$  is achieved uniquely at $b=\hat{b}$.
\end{lemma}
\begin{proof}
We have
\begin{equation}\label{epart}
\frac{\partial}{\partial b} Q(b) = 4\beta b - \ln\frac{t+2b}{t-2b}.
\end{equation}
For~\eqref{epart} to be zero we need to have
$$\beta t = \frac{1}{2} \frac{t}{2b} \ln\frac{1+2b/t}{1-2b/t}.$$
We will show that the function
\begin{equation}\label{eeo}
y\mapsto \frac{1}{y}\ln\frac{1+y}{1-y}
\end{equation}
is increasing on $[0,1]$ and hence there is unique critical $b$ (and we chose $\beta$ so that $\hat{b}$ is critical).

The derivative of~\eqref{eeo} is
$$
\frac{2}{y(1-y^2)} - \frac{1}{y^2}\ln\frac{1+y}{1-y},
$$
which we will show is positive on $(0,1)$. It is sufficient to show that the following is increasing
\begin{equation}\label{eeozz}
\frac{2y}{(1-y^2)} - \ln\frac{1+y}{1-y},
\end{equation}
since the value at $y=0$ is $0$. The derivative of~\eqref{eeozz} 
$$
\frac{4y^2}{(1-y^2)^2}>0,
$$
true for $y\in (0,1)$.
\end{proof}

Let $\hat{u}\in (\hat{b},t)$. We are going to take $\gamma$ such that
\begin{equation}\label{e2}
\beta + 3\gamma = \frac{1}{4\hat{u}}\ln \frac{t + 2\hat{u}}{t - 2\hat{u}}.
\end{equation}
Note that since~\eqref{eeo} is an increasing function on $(0,1)$ and we chose $\hat{u}>\hat{b}$ we have $\gamma>0$.

The parametrization of $\gamma$ is chosen in such a way that $\hat{u}$ will bound the infinity norm of the
(local) maximizers of $\Phi(\mathbf{b})$; the following lemma makes this formal.

\begin{lemma}\label{lo}
Let $\beta,\gamma$ be given by~\eqref{e1} and~\eqref{e2}. Let $\{a_v\}_{v\in V}$ be such that $a_v\in\{\pm 1\}$ for each $v\in V$. Suppose $\mathbf{b}=\{b_v\}_{v\in V}$ maximizes $\Phi(\mathbf{b})$ subject to $a_v b_v\geq 0$ for $v\in V$ (that is the signs of $b_v$'s are given by $a_v$'s).
Then
\begin{equation}\label{eupo}
\max_{v\in V} |b_v|\leq\hat{u}.
\end{equation}
\end{lemma}

\begin{proof}
Let $\mathbf{b}=\{b_v\}_{v\in V}$ be a maximizer.  W.l.o.g. $b_1$ has the largest absolute value and $b_1>0$. Let
\begin{equation}\label{ES}
S=\sum_{u; \{u,1\}\in E} b_u.
\end{equation}
Note that $S\in [-3b_1,3b_1]$ (since $G$ is a $3$-regular graph and $b_1$ bounds other $|b_u|$'s). Viewing $\Phi$ as a function of $b_1$ we have
\begin{equation}\label{PB1}
\Phi(b_1) = Q(b_1) - 4\gamma S b_1 + R,
\end{equation}
where $R$ and $S$ do not depend on $b_1$. Since $\mathbf{b}=\{b_v\}_{v\in V}$ is a maximizer and $b_1$ is not on the boundary
we must have
\begin{equation}\label{PB1de}
\frac{\partial}{\partial b_1} \Phi(b_1)  = 4\beta b_1 + 4\gamma S - \ln\frac{1+2b_1/t}{1-2b_1/t} = 0.
\end{equation}
We have
$$
\frac{1}{4 b_1} \ln\frac{1+2b_1/t}{1-2b_1/t} = \beta + \gamma S/b_1 \leq \beta + 3\gamma = \frac{1}{4\hat{u}}\ln \frac{1 + 2\hat{u}/t}{1 - 2\hat{u}/t},
$$
and since~\eqref{eeo} is increasing we conclude $b_1\leq\hat{u}$. 
\end{proof}

\begin{proof}[Proof of Theorem~\ref{TMAIN}.]
We will set
$$
t = n^C, \quad\quad \hat{b}=t^{3/4+\delta}\quad\mbox{and}\quad \hat{u} = \hat{b} + t^{3/4+\delta'} 
$$
where $C>0$ is large and $\delta>\delta'>0$ are small; we will in fact take $\delta=\eps/6$, $\delta'=\eps/12$, $C=\lceil 3/\eps\rceil$ and $c=1/(C+1)$, but this will only be relevant later. The value of $\beta$ is given by~\eqref{e1} and the value of $\gamma$ is given by \eqref{e2}.
For $x\leq 1/2$ we have
$$
(1/x) \ln \frac{1+x}{1-x}\leq 2(1 + x^2);
$$
this follows, by integration, from $2/(1-x^2)\leq 2(1+x)(1+3x)$ for $x\in [0,1/2]$. Hence, for sufficiently large $n$, we may apply this for $x=2\hat{b}/t\leq 1/2$ (by our choice of $t$ and $\hat{b}$), to obtain that
\begin{equation}\label{ALUPE}
\beta = \frac{1}{2t} (1/(2\hat{b}/t)) \ln \frac{1+2\hat{b}/t}{1-2\hat{b}/t}\leq \frac{1}{t} + 8 t^{-3/2+ 2\delta}.
\end{equation}
Plugging the definitions of $\hat{b}$ and $\hat{u}$ into~\eqref{e2} we obtain
\begin{equation}\label{GASY}
\gamma = \frac{8}{9} t^{-3/2 + \delta + \delta'} + o( t^{-3/2 + \delta + \delta'}).
\end{equation}

We next give bounds on the partition function $Z_J$ depending on the size of the cut of~$G$. Recall that 
\begin{equation*}\tag{\ref{eq:partitionfun}}
Z_J = K\sum_{\{b_v\}_{v\in V}} \bigg( \prod_{v\in V} {t\choose t/2 + b_v}\bigg) \exp\Big( 2\beta \sum_{v\in V} b_v^2 - 4 \gamma \sum_{\{u,v\}\in E} b_u b_v \Big),
\end{equation*}
where the sum ranges over all vectors $\mathbf{b}=\{b_v\}_{v\in V}$ with integer entries in $[-t/2,t/2]$.

Suppose first that $G$ has a max-cut of size $A$, and fix a particular max-cut of $G$. Recall that we assume $A\geq \tau|E|/2\geq (3/4) n$. The 
Ising configurations on $G_t$ that correspond to the max-cut contribute at least the following to the partition function 
\begin{equation}\label{COMAX}
T_1 = K{t\choose t/2+\hat{b}}^n \exp\Big( 2\beta n \hat{b}^2 + 4 \gamma  \hat{b}^2 (2 A - (3/2)n)\Big).
\end{equation}
Now suppose that all cuts of $G$ have size at most $ A/\tau$. Recall from \eqref{ZOB}  that $Q(b)\geq  2\beta b^2 + \log {t\choose t/2 + b}$. The partition function is at most
\begin{equation}\label{COTAU}
T_2 = K(t+1)^n \exp\Big(n Q(\hat{b})+ 4 \gamma  \hat{u}^2 (2 A/\tau - (3/2)n)\Big),
\end{equation}
since the number of choices for $\{b_v\}_{v\in V}$ in \eqref{eq:partitionfun} is bounded by $(t+1)^n$ (the sign of the $b_i$'s 
determines the corresponding cut),  and  a contribution of any choice of $\mathbf{b}=\{b_v\}_{v\in V}$ to the partition function is bounded by
\begin{equation}\label{ezz}
K\exp\Big(n Q(\hat{b})+ 4 \gamma  \hat{u}^2 (2 A/\tau - (3/2)n)\Big).
\end{equation}
To see why the last claim is true, consider a vector $\mathbf{b}=\{b_v\}_{v\in V}$  that makes the 
largest contribution to the partition function. For every $i$ the sign of $b_i$ is the opposite of the sign 
of $\sum_{u;\{u,i\}\in E}b_u$; otherwise flipping the sign of $b_i$ we increase the contribution. By
Lemma~\ref{lo}, in the optimal choice for $\mathbf{b}=\{b_v\}_{v\in V}$, all have absolute value less than $\hat{u}$. 
Hence  the upper bound in~\eqref{ezz} applies, since $nQ(\hat{b})$ is an upper bound for the 
$\sum_{v\in V} Q(b_v)$  term (using Lemma~\ref{lem:maxb} on each term) and setting all $b_v$'s to $\hat{u}$ is an upper bound for the remaining term.

To finish the proof, it remains to show that $\log T_1-\log T_2\geq 2N^c$ and that the spectral diameter of $J$ is less than $1+N^{-1/2+\varepsilon}$. We start with the first inequality.  We will use the following lower bound on the binomial coefficient
\begin{equation}\label{BICO}
\ln {t\choose t/2+b} \geq t H(1/2 + b/t) - \ln(t+1).
\end{equation}
Using~\eqref{BICO}  and the expression for $T_1$ in \eqref{COMAX}, we have
\begin{equation*}
\log T_1 \geq \log K+ n Q(\hat{b}) - n\log(t+1) + 4 \gamma  \hat{b}^2 (2 A - (3/2)n).
\end{equation*}
Combining with the expression for $T_2$ from \eqref{COTAU} we have
\begin{equation}\label{EDI}
\log T_1 - \log T_2 \geq -2n(\log(t+1)) + 4\gamma \left( 
\hat{b}^2 (2 A - (3/2)n) - \hat{u}^2 (2 A/\tau - (3/2)n) \right).
\end{equation}
We can bound
\begin{equation*}
\begin{split}
& \hat{b}^2 (2 A - (3/2)n) - \hat{u}^2 (2 A/\tau - (3/2)n)
\geq 2A ( \hat{b}^2 - \hat{u}^2/\tau)\\
& = 2A \hat{b}^2 (1-1/\tau) + 2A(\hat{b}^2 - \hat{u}^2)/ \tau
\geq C' n \left( t^{3/2 + 2\delta} + O(t^{3/2+\delta+\delta'}) \right),
\end{split}
\end{equation*}
where $C'$ is a constant depending on $\tau$. Combining with~\eqref{GASY}
we have
$$
4\gamma \left(
\hat{b}^2 (2 A - (3/2)n) - \hat{u}^2 (2 A/\tau - (3/2)n) \right) \geq C' n 
\big( t^{3\delta+\delta'} + o(t^{3\delta+\delta'}) \big),
$$
which implies from~\eqref{EDI} that, for large $n$, $\log T_1 - \log T_2\geq 2n\geq 2N^{c}$ (recall $N=tn=n^{C+1}$ and  $c=\tfrac{1}{C+1}$).

It remains to bound the spectral diameter of the interaction matrix $J$. We can write $J$  as $J_1+J_2$ where $J_1$ is a block diagonal matrix with 
$n$ blocks of size $t\times t$ filled with $\beta$'s and $J_2$ is a matrix with $(v,i)$ entry 
$(u,j)$ entry equal to $-\gamma$ if $\{u,v\}\in E$ (and the remaining entries are zero).
We have, using~\eqref{ALUPE}
\begin{equation*}
\|J_1\|_2 = t\beta \leq 1 + 8t^{-1/2+2\delta}.
\end{equation*}
Note that $J_1$ is a positive semidefinite matrix and hence has spectral diameter bounded  by the 2-norm
(recall that spectral diameter of a symmetric matrix is the difference between the largest and the smallest eigenvalue). We also have, using~\eqref{GASY}, that
\begin{equation*}
\|J_2\|_2 \leq 3t\gamma = (24/9) t^{-1/2+\delta+\delta'} + o(t^{-1/2+\delta+\delta'})
\end{equation*}
and hence the spectral diameter of $J_2$ is bounded by $(48/9) t^{-1/2+\delta+\delta'} + o(t^{-1/2+\delta+\delta'})$.
The two bounds combined yield  that the spectral diameter of $J=J_1+J_2$ is bounded by 
$$1 + 8t^{-1/2+2\delta} + o(t^{-1/2+2\delta}).$$

The matrix $J$ has dimensions $N\times N$ where $N=tn = n^{C+1}$. Hence, the bound
on the spectral diameter of $J$ in terms of $N$ is 
$$
1 + N^{-1/2+2\delta+1/(2(C+1))} + o(N^{-1/2+2\delta+1/(2(C+1))}).
$$
Since $\delta=\eps/6$, $\delta'=\eps/12$, and $C=\lceil 3/\eps\rceil$ we obtain that the spectral diameter of $J$ is $\leq 1+N^{-1/2+\varepsilon}$, as needed.
\end{proof}

\newpage

\bibliographystyle{alpha}
\bibliography{biblio}

\end{document}